\documentclass[aps, pra,floats,tightenlines, a4paper,twocolumn,10pt,superscriptaddress]{revtex4}
\usepackage{pstricks}%
\usepackage{graphicx}
\usepackage{dcolumn}
\usepackage{bm}
\usepackage{dsfont}
\usepackage{amsmath,amssymb}
\usepackage{pifont}
\usepackage{psfrag}
\usepackage[latin2]{inputenc}
\usepackage{bbm}
\usepackage{soul}
\usepackage{hyperref}

\newcommand\Tr{\text{Tr\,}}

\newcommand{\BE}{\begin{equation}}
\newcommand{\EE}{\end{equation}}

\newcommand{\skipc}[2]{}

\newcommand{\fig}[1]{Fig.~\ref{#1}}
\newcommand{\eq}[1]{Eq.~(\ref{#1})}
\newcommand{\Sec}[1]{Sec.~\ref{#1}}

\newcommand{\lem}[1]{Lemma~\ref{#1}}

\newtheorem{theorem}{Theorem}
\newtheorem{lemma}[theorem]{Lemma}

\newcommand{\qed}{\nobreak \ifvmode \relax \else
      \ifdim\lastskip<1.5em \hskip-\lastskip
      \hskip1.5em plus0em minus0.5em \fi \nobreak
     $\square$\fi}

\newenvironment{proof}[1][Proof.]{\begin{trivlist}
\item[\hskip \labelsep {\bfseries #1}]}{\end{trivlist}}
\newenvironment{definition}[1][Definition]{\begin{trivlist}
\item[\hskip \labelsep {\bfseries #1}]}{\end{trivlist}}

\newcommand{\be}{\begin{equation}}
\newcommand{\ee}{\end{equation}}
\newcommand{\eea}{\end{eqnarray}}
\newcommand{\bea}{\begin{eqnarray}}

\newcommand{\ket}[1]{\ensuremath{|#1\rangle}}
\newcommand{\bra}[1]{\ensuremath{\langle#1|}}
\newcommand{\braket}[2]{\ensuremath{\langle #1|#2\rangle}}

\usepackage{color}

\begin{document}

\title{Quantum-accessible reinforcement learning beyond strictly epochal environments}

\author{A. Hamann}
 \affiliation{Institut f\"ur Theoretische Physik, Universit\"at Innsbruck, Technikerstra{\ss}e 21a, 6020 Innsbruck, Austria}

\author{V. Dunjko }
 \affiliation{LIACS,  Leiden  University,  Niels  Bohrweg  1,  2333  CA  Leiden,  The  Netherlands}

\author{S. W\"olk}
 \affiliation{Institut f\"ur Theoretische Physik, Universit\"at Innsbruck, Technikerstra{\ss}e 21a, 6020 Innsbruck, Austria}
 \affiliation{Institute of Quantum Technologies, German Aerospace Center (DLR), D-89077 Ulm, Germany}

\date{\today}

\begin{abstract}

In recent years, quantum-enhanced machine learning has emerged as a particularly fruitful application of quantum algorithms, covering aspects of supervised, unsupervised and reinforcement learning. Reinforcement learning offers numerous options of how quantum theory can be applied, and is arguably the least explored, from a quantum perspective.  
Here, an agent explores an environment and tries to find a behavior optimizing some figure of merit.  Some of the first approaches investigated settings where this exploration can be sped-up, by considering quantum analogs of classical environments, which can then be queried in superposition. If the environments have a strict periodic structure in time (i.e. are strictly episodic), such environments can be effectively converted to conventional oracles encountered in quantum information. 
However, in general environments, we obtain scenarios that generalize standard oracle tasks. In this work we consider one such generalization, where the environment is not strictly episodic, which is  mapped to an oracle identification setting with a changing oracle. We analyze this case and show that standard amplitude-amplification techniques can, with minor modifications, still be applied to achieve quadratic speed-ups, and that this approach is optimal for certain settings. This results constitutes one of the first generalizations of quantum-accessible reinforcement learning.

\end{abstract}


\maketitle


\section{Introduction}

In the last few years, there has been much interest in combining quantum computing and machine learning algorithms. In the domain of quantum-enhanced machine learning, the objective is to utilize quantum effects to speed-up, or otherwise enhance the learning performance. The possibilities for this are numerous \cite{Dunjko2018}. E.g. variational circuits can be used as a type of ``quantum neural network''  (more precisely, using them as function approximators  which cannot be evaluated efficiently on a conventional computer), which can be trained as a supervised learning (classification) \cite{Havlicek2019, farhi2018classification} or unsupervised learning model (generative models) \cite{Aimeur2013}. There also exist various approaches where algorithmic bottlenecks of classical algorithms are sped-up, via annealing methods \cite{neven2018}, quantum linear-algebraic methods \cite{PhysRevLett.103.150502}, or via sampling enhancements \cite{Dunjko2016}.
If the data is assumed to be accessible in a quantum form (``quantum database'') then anything from polynomial, to exponential speed-ups of classical algorithms may be possible \cite{Biamonte2017,Dunjko2018,Chia2019,Gyurik2020}\footnote{In recent times, due to progress in quantum-inspired algorithms,  the domain of algorithms where exponential speed-ups are to be expected has reduced, but many possibilities for classically intractable computations still exist.}.

Modern reinforcement learning (RL), an interactive mode of learning, combines aspects of supervised and unsupervised learning, and consequently allows a broad spectrum of possibilities how quantum effects could help.

In RL \cite{Sutton1998,Russell2003,Briegel2012}, we talk about a learning agent which interacts with an environment, by performing actions, and perceiving the environmental states, and has to learn a ``correct behavior'' -- the optimal policy -- by means of a feedback rewarding signal. Unlike a stationary database, the environment has its own internal memory (a state), which the agent alters with its actions.

In quantum-enhanced RL, we can identify two basic scenarios: i) where quantum effects can be used to speed up the internal processing \cite{PhysRevX.4.031002,Jerbi2019}, and the interaction with the environment is classical, and ii) where the interaction with the environment (and the environment itself) is quantum.  
The first framework for such ``quantum-accessible'' reinforcement learning modeled the environment as a sequence of quantum channels, acting on a communication register, and the internal environmental memory -- this constitutes a direct generalization of an unknown environment as a map-with-memory (other options are discussed shortly). In this case, the action of the environment cannot be described as unitary mapping, without considering the entire memory of the environment. In general, this memory is inaccesible to the agent. However, as discussed in \cite{Dunjko2016}, under the assumptions that the environmental memory can be purged or uncomputed in pre-defined periods, such blocks of interaction do become a (time-independent) unitary, and amenable to oracle computation techniques. 
For instance, in \cite{Dunjko2016} it was shown that  the task of identifying a sequence of actions which leads to a first reward (a necessary step before any true learning can commence) can be sped up using quantum search techniques, and in \cite{Dunjko2017} it was shown how certain environments encode more complex oracles -- e.g. Simon's oracle and Recursive Fourier Sampling oracles, leading to exponential speed-ups over classical methods.

For the above techniques to work, however, the purging of all of environmental memory is necessary to achieve time-independent unitary mappings. However, real task environments are typically not (strictly) episodic, motivating the question of what can be achieved in these more general cases.  Here, we perform a first step towards generalization by considering environments where the length of the episode can change, but this is signaled and the estimate of the episode lengths are known.
This RL scenario is well-motivated, and, fortunately maps to an oracle identification problem where the oracles change. While this generalizes standard oracular settings, it is still sufficiently simple such that we can employ standard techniques (essentially amplitude amplification) and prove the optimality of our strategies  in certain settings.

The paper is organized as follows. We will first summarize the basics scenario of quantum-accessible reinforcement learning in \Sec{sec:QRL} and discuss 
the mappings from constrained (episodic) RL scenarios to oracle identification. We show how this must be generalized for more involved environments, prompting our definition of the ``changing oracle'' problem stemming from certain classes of RL environments.
   In \Sec{sec:changing_oracle}, we focus on the changing oracle problem, analyze the main regimes, and provide an upper bound for the average success probability for the case of monotonically increasing winning space in \Sec{sec:upper_bound}. We proof in \Sec{sec:Grover} that performing consecutive Grover iterations saturates this bound. We then discuss the more general case of only overlapping winning spaces in \Sec{sec:Grover2}. We conclude by summarizing our results, by discussing possible extensions and by noting on the implications of our results of the changing oracle problem for QRL  in \Sec{sec:Conclusion}.


\section{Quantum-accessible Reinforcement Learning\label{sec:QRL}}

RL can be described as an interaction  of a learning agent $A$ with a task environment $E$ via the exchange of messages out of a discrete set which we call actions $\mathcal{A}=\lbrace a_j\rbrace$ (performed by the agent) and percepts $\mathcal{S}=\lbrace s_j\rbrace$ (issued by the environment). In addition, the environment also issues a scalar reward $\mathcal{R}=\lbrace r_j\rbrace$, which informs the agent about the quality of the previous actions  and can be defined as being a part of the percepts. The goal of the agent is to receive as much reward as possible in the long term. 

In theory of RL, the most studied environments are exactly describable by a Markov Decision process (MDP). An MDP is specified by a transition mapping $T(s' | a, s)~\in~\mathbbmss{R}^{\geq 0}$, and a reward function $R(s,a)~\in~\mathbbmss{R}$. The transition mapping $T$  specifies the probability of the environment transiting from state $s$ to $s'$, provided the agent performed the action $a$, whereas the reward function assigns a reward value to a given action of an agent in a given environmental state.

Note that in standard RL, the agent does not have a direct access to the mapping $T$, but rather to learn it, it must explore, i.e. to act in the environment which is governed by $T$. On the other hand, in dynamical programming problems (intimately related to RL), one often assumes access to the functions $T$ and $R$ directly. This distinction leads to two different takes on how agent-environment interaction can be quantized. 

In recent works  \cite{Cornelissen2018,Neukart2018,Ronagh2017} coherent access to the transition mapping $T$ is assumed, in this case, lower quantum bounds for finding the optimal policy have been found~\cite{Ronagh2019}.

In this paper, we consider the other class of generalization, proposed first in \cite{Dunjko2016}.
Here, the agent-environment interaction is modeled as a communication between an agent (A) and the environment (E) over a joint communication channel (C), thus in a three-partite Hilbert space $\mathcal{H}_E\otimes \mathcal{H}_C \otimes \mathcal{H}_A $, denoting the memory of the environment, the communication channel, and the memory of the agent. The two parties A and E interact with each other by  performing alternately completely positive trace preserving (CPTP) maps on their own memory and the communication channel. Different AE combinations are defined as  equivalent in the classical sense, if their interactions are equivalent under constant measurements of C in the computational basis. For classical luck favoring AE settings with a deterministic  strictly epochal environment E it is possible to create  a classical equivalent quantum version A$^q$E$^q$ which outperforms AE in terms of a given figure of merit as shown in \cite{Dunjko2016}.

\subsection{Strictly epochal environments}

This can be achieved by slightly modifying the maps as to purge the environmental memory which couples to the overall interaction preventing a unitary time evolution of the agents memory.  A detailed discussion of this procedure and necessary condition on the setting are outlined in \cite{Dunjko2016}. However, for our setting it is sufficient that the interaction of the agent with the environment can be effectively described as  oracle queries. Specifically if environments are strictly episodic, meaning after some fixed number of steps the setting is re-set to an initial condition, then the environmental memory can be uncomputed, or released to the agent at the end of an epoch. With this modification (called memory scavenging and hijacking in earlier works), blocks of interactions effectively act as one, time-independent unitary, which can be queried using standard quantum techniques to obtain an advantage. In this scenarios, it is possible to summarize the effect of the environment on the  state $\ket{\vec{a}}=\ket{a_1,\cdots,a_M}$ describing the sequence of actions played during a complete epoch of length $M$ by an oracle
\BE
O \ket{\vec{a}} = \left\lbrace \begin{array}{cc} -\ket{\vec{a}}&\,\text{if}\,\vec{a} \in W \\
\ket{\vec{a}}&\text{else}\end{array}\right.
\EE
 with $W$ denoting the winning space containing all sequences of actions of length $M$ which obtained a reward $r(\vec{a})$ larger than a predefined limit. Then, the agent can prepare an equal superposition state of all possible action sequences 
\BE 
\ket{\psi}=\frac{1}{\sqrt{N}}\sum\limits_{i=1}^N\ket{\vec{a}_i}
\EE
with typically $N=|\mathcal{A}|^M$. Consecutively, it can perform amplitude amplification by e.g. performing consecutive Grover iterations \cite{Grover1997,Grover1998,Brassard2000} by applying the unitary
\BE
G_\psi=\left(\mathds{1}-2\ket{\psi}\bra{\psi}\right)O\label{eq:U_G}
\EE
to $\ket{\psi}$. In this way, the agent can increase the probability to find a first winning sequence which increases in luck-favoring settings also the probability to be rewarded in the future. This  approach leading to a quadratic speed-up in exploration can be applied to many settings. However, also super-polynomial or exponential improvements can be generated for special RL settings \cite{Dunjko2017}.

\subsection{Beyond strictly epochal environments}

The simplest scenario of task environments which cannot be reformulated as an oracular problem, are arguably those which involve two oracles.  We will consider this slight generalization  in this work, as it still allows for a relatively simple treatment. This setting includes environments which simply change as a function of time such as e.g reinforcement learning for managing power consumption or  channel allocation in cellular telephone systems \cite{Han2018,Teasuro2008,Silva2006,Singh1996}. If the instances of change are known, again the blocking is possible, in which case we obtain the setting where we can realize access to an oracle but which changes as a function of time.
Closely related to this is a more specific case of variable episode length. This setting, although more special, is in particular interest in RL. Episodic environments are usually constructed by taking an arbitrary environment, and establishing a cut-off after a certain number of steps. The resulting object is again an environment derived from the initial setting. This construction is special in that given any sequence of actions $\vec{a}$ which is rewarding in a derived environment with cut-off after  $m$ steps, any sequence of actions in the environment which has a larger cut  off $M>m$ which has $\vec{a}$ as a prefix is rewarded in the second. An example for such an environment is the Grid-world problem which consists in navigating a maze and the task is to find a specific location that is rewarded \cite{Russell2003, Sutton1998,Melnikov2018}.

The classical scenarios described above, under oraculization techniques map onto the changing oracle problem (described in detail in the following section) where at a given time an oracle $\tilde{O}$ is exchanged by a different oracle ${O}$.  This generalization especially captures the scenario of a single increment of an epoch length  from $m$ to $M>m$  for  search in QRL. In this special case,  the winning space $\tilde{W}$ of $\tilde{O}$ is a subspace of $W$ of $O$. We will proof that  the optimal algorithm in this case is given by a Grover search with a continuous coherent time evolution using both oracles consecutively. However, continuing the coherent time evolution of a Grover search can be suboptimal when $\tilde{W}\not\subset W$.    
The arguments following in the next section can be used iteratively to describe multiple changes/increments of the winning space.


\section{The changing oracle problem\label{sec:changing_oracle}}
 
The situation above can be abstracted as a ``changing oracle'' problem which we specify here.
As suggested, we consider an ``oracle'' to be a standard phase-flip oracle, such that $ O\ket{x} = (-1)^{f(x)}\ket{x}$, where $f: X \rightarrow \{0,1\}$ is a characteristic function  on a set of elements $X$, with $|X| = N$; in our case $X$ denotes sequences of actions of some prescribed length. The winning set is denoted  by $W = \{x \in X|  f(x) = 1 \}$, and the states $\ket{x}$ denote a (known) orthonormal basis.

In the changing oracle problem, we consider two oracles $\tilde{O}$, and  ${O}$, with respective winning sets $\tilde{W}$ and $W$. The problem specifies two time intervals, phases, in which only one of the two oracles is available: time-steps $1 \leq k \leq K$ during which only access to  $\tilde{O}$ is available, and time-steps $K+1 \leq k < K+J$ during which only access to the second oracle $O$ is available.

For simplicity, we assume that the values of $K$, $J$, $N$ as well as the sizes of the winning sets $|\tilde{W}|=\tilde{n}$ and $|W|=n$ are known in advance, and in general, the objective is to either output an $x \in \tilde{W}$ before $K$, or, to output $x \in {W}$ in the remainder of the time.  We will refer to both $x$ as the solution. However, the exact time when the oracle changes, and does $K$ and $J$, is not important and can be unknown as we show later. Unless $K$ is in $\Omega(\sqrt{N/\tilde{n}})$, in general attempts to find a solution in the first phase will have a very low success probability no matter what we do due to the optimality of Grover's search. However, even in this case, having access to $\tilde{O}$ in the first phase, may improve our chances to succeed in the second. This is the setting we consider.

The optimal strategies vitally depend on the known relationship between $W$ and $\tilde{W}$. We will first briefly discuss all possible setting before focusing on the most interesting cases. Note, in this paper we are not looking for a strategy which uses a minimal number of queries until a solution is found, but rather, a strategy which maximizes the success probability for a fixed number of queries. However, it is also known that Grover's search achieves the fastest increase of success probability \cite{Zalka1999}. Therefore, the here described algorithms can be also used to optimize the number of queries. However, the corresponding figure of merit, which needs to be optimized, has to be defined precisely for such tasks.

\paragraph{} In the worst case, there may be no known correlation between $W$ and $\tilde{W}$. In this case, we have no advantage from having access to $\tilde{O}$, and the optimal strategy is a simple Grover's search in the second phase. 

\paragraph{} Another case with limited interest is when $W$ and $\tilde{W}$ are known to be disjointed. In this case, the first oracle might be used to constrain the search space to the complement $\tilde{W}^c,$ which contains $W$. The lower bounds for this setting are easy to find: we can assume that at $K$ the set $\tilde{W}$ is made known (any state we could have generated using $\tilde{O}$ can be generated with this information). However, in this case, the optimal strategy is still to simply apply quantum search over the restricted space $\tilde{W}^c$ if it can be fully specified. But since we most often encounter cases where  $\tilde{n}=|\tilde{W}|$ is (very) small compared to $N$, the improvement that could be obtained is also minor.
\paragraph{} Similar reasoning follows also  when the sets are not disjoint, but the intersection is small compared not just to $N$, but to $|W|$ and $|\tilde{W}|$.   In this case, again we can find lower bounds by assuming that the non-overlapping complement becomes known. In addition, we assume that we can prepare any quantum state, which has an upper bound on the overlap with any state corresponding to the intersection, $x \in W \cap \tilde{W}$. Then, the optimal strategy is again governed by the optimality of Grover-based amplitude amplification \footnote{More generally, we can allow only states for which, under no quantum channel, allow us to determine such $x$ with probability better than given by Grover iterations. This setting is a bit more involved, but it should be clear that as long as this probability is very small, whatever we do in the next phase, cannot be much better than starting from scratch.}

This brings us to the situations which are more interesting, specifically, when the overlap $W_a=W \cap \tilde{W}$ is large (see Appendix \ref{app:large overlap} for exact definition).  
 
Due to our motivation stemming from aforementioned RL settings, we are particularly interested in the case when $\tilde{W} \subseteq {W},$ for which we give the optimal strategy, which turns out to be essentially Grover's amplification where we ``pretend'' that the oracles hadn't changed.

The other cases,  ${W} \subseteq \tilde{W}$, and the more generic case where the overlap is large, but no containment hold are less interesting for our purpose, so we briefly discuss the possible strategies without proofs of optimality.

\subsection{ Increasing winning spaces: upper bound on average final success probabilities\label{sec:upper_bound}}
In the following, we consider the above described changing oracle problem with monotonically increasing winning spaces $\tilde{W} \subseteq {W}$ and derive upper bounds for the maximal average success probability $p_{K+J}$ of finding an element $x\in W$ at the end of the second phase. 
The changing oracle problem is outside the standard settings for which various lower bounding techniques have been developed \cite{Arunachalam2019, Ambainis2002,Ambainis2006}, but the setting is simple enough to be treatable by modifying and extending  techniques introduced to lower bound unstructured search problems \cite{Zalka1999}.

To find lower bounds, we first prove that we can restrict our search for optimal strategies to averaged strategies as defined in Appendix \ref{app:averaged_strategy}. This induces certain symmetries which restrict the optimization to an optimization of two angles $\alpha$ and $\Delta$, one for each phase. Finally we derive bounds $\alpha(K)$ and $\Delta(J)$ for these angles depending on $K,J$ which in turn restrict the optimal success probability $p_{K+J}$.

The search for an optimal strategy can be limited to strategies based on pure states and unitary time evolutions since it is possible to  purify any search strategy by going from a Hilbert space $\mathcal{H}_A$ spanned by $\lbrace \ket{x}\rbrace$ into a larger Hilbert space $\mathcal{H}_{AB}=\mathcal{H}_A\otimes\mathcal{H}_B$. As a consequence, every search strategy $T=(\lbrace U_k\rbrace,\ket{\psi(0)})$ based on $K+J$ oracle queries can be described by a set of $K+J$ unitaries $U_k$ and initial state $\ket{\psi(0)}$. Our knowledge about possible winning items after $k$ oracles queries is then encoded in the quantum state
\BE
\ket{\psi(k)}=U_kO_k \cdots U_1O_1\ket{\psi(0)}
\EE
with $O_k=\tilde{O}$ for $1\leq k\leq K$ and $O_k={O}$ for $K+1\leq k\leq J$. The success probability at the end of the second phase is then given by
\BE
p_{K+J}=\Tr[P_\mathcal{W}\ket{\psi(K+J)}\bra{\psi(K+J)}]
\EE
with 
\BE
P_\mathcal{W}=\Big(\sum\limits_{x\in \mathcal{W}}\ket{x}_A\bra{x}\Big) \otimes \mathds{1}_B.
\EE
Our goal is to maximize the success probability $p_{K+J}$ average over all possible functions $\tilde{f}(x)$ and $f(x)$ with fixed sizes of the winning spaces $|\tilde{W}|=\tilde{n}$ and $|{W}|={n}\geq \tilde{n}$. Different realization of $\tilde{f}(x)$ and $f(x)$ can be generated by substituting all oracle queries $O_k$ by $\sigma O_k \sigma^\dagger$ and the projector $P_\mathcal{W}$ by $\sigma P_\mathcal{W}\sigma^\dagger$ where $\sigma$ denote a permutation operator acting on $\mathcal{H}_A$. As a consequence, an optimal strategy is a strategy T which maximizes 
\BE
\bar{p}_T=\frac{1}{N!}\sum\limits_{\sigma \in \Sigma_A}p_T(\sigma) 
\EE
with
\begin{eqnarray}
p_T(\sigma)&=&\Tr\Big[\sigma P_\mathcal{W}\sigma^\dagger \ket{\psi(k,\sigma)}\bra{\psi(k,\sigma)}\Big]\\
\ket{\psi(k,\sigma)}_{AB}&=&U_k\sigma O_k\sigma^\dagger\cdots U_1 \sigma O_1\sigma^\dagger \ket{\psi(0)}_{AB}\end{eqnarray}
at the end of the second phase such that $k=K+J$. Here, $\Sigma_A$ denotes the set of all possible permutations in $\mathcal{H}_A$.

We can further limit the search for optimal strategies to averaged strategies $\bar{T}$  as defined Appendix \ref{app:averaged_strategy}   because
\begin{lemma}
The success probability $p_{\bar{T}}(\sigma)$ of the averaged strategy $\bar{T}$ is equal to the average success probability $\bar{p}_T$ of the strategy $T$ for every permutation $\sigma\in\Sigma_A$. 
\label{lemma:sym_strat}
\end{lemma}
as proven in Appendix \ref{app:averaged_strategy}. In the following, we consider only average strategies such that $p=\bar{p}$ and therefore omit the "bar" denoting an average value. 

In addition, these strategies leads to symmetry properties of the unitaries $U_k$ and resulting states $\psi(k)$ under permutations $\sigma$ as outlined in detail in Appendix \ref{app:averaged_strategy}). These symmetry properties will limit the optimization overall strategies to an optimization of a few parameters or angles as we will outline below. These parameters are then again  upper bounded by the optimality of Grover search.  

Due to the above mentioned symmetry properties, we can write the state $\ket{\psi}$ at the end of the first phase via (see Appendix \ref{app:symmetries})
\BE
\ket{\psi(K)}= \cos \varepsilon \ket{\phi_s}+ \sin \varepsilon \ket{\phi_\perp}\label{eq:phi_perp}
\EE
with the symmetric component
\BE
\ket{\phi_s}=\sin \phi \ket{w_s}+\cos \phi \ket{\ell_s}\\
\EE
and a component 
\BE
\ket{\phi_\perp}=\ket{w_\perp}
\EE
orthogonal to $\ket{\phi_s}$ which contain for $\tilde{\mathcal{W}}\subseteq \mathcal{W}$ only winning items.
The normalized components $\ket{w_s}$ contains only winning items and $\ket{\ell_s}$ only losing items according to the second oracle $O$.  The angles $\varepsilon$ and $\phi$ are parameters depending on the strategy performed during the first phase. Their values are bounded by the success probability at the end of the first phase given by
\BE
p_K=\cos^2\varepsilon\sin^2\phi+\sin^2\varepsilon. \label{eq:p_K}
\EE

The time evolution during the second phase described by
$
V=U_{K+J}O\cdots U_{K+1}O
$
is also symmetric and thus transforms the symmetric component $\ket{\phi_s}$ into a symmetric component and $\ket{w_\perp}$ into a component orthogonal to $V\ket{\phi_s}$. As a consequence, the final success probability $p_{K+J}$ can be divided into 
\BE
p_{K+J}=\cos^2 (\varepsilon)\; p_s+ \sin^2(\varepsilon) \;p_\perp
\EE
with (see Appendix \ref{app:symmetries})
\begin{eqnarray}
p_s&=&\Tr\left[P_\mathcal{W}V\ket{\phi_s}\bra{\phi_s}V^\dagger \right]\\
p_\perp&=&\Tr\left[P_\mathcal{W}V\ket{w_\perp}\bra{w_\perp}V^\dagger \right].
\end{eqnarray}

The winning probability $p_\perp$ of the orthogonal part is maximal if $p_\perp=1$ which can be achieved if e.g. $V$ acts on $\ket{w_\perp}$ as identity. By writing the winning probability of the symmetric part via 
$
p_s=\sin^2(\phi+\Delta)
$
we can quantify the final success probability via

\begin{eqnarray}
p_{K+J}&\leq&\cos^2 (\varepsilon) \sin^2(\phi+\Delta)+\sin^2(\varepsilon)\\
&\leq& 1-\cos^2(\varepsilon)\cos^2\left(\phi+\Delta\right).
\end{eqnarray}
With the help of \eq{eq:p_K} we can rewrite $\cos^2 \varepsilon$ via
$
\cos^2\varepsilon=(1-p_K)/\cos^2 \phi
$
leading to
\BE
p_{K+J}\leq 1-(1-p_K)\frac{\cos^2(\phi+\Delta)}{\cos^2 \phi}\label{eq:p2}.
\EE
As a consequence, $p_{K+J}$ is monotonically increasing with $p_K,\phi,\Delta$ provided $0\leq \phi\leq \pi/2$ and $0\leq \phi+\Delta\leq \pi/2$. Thus an optimal strategy optimizes $p_K$ and $\phi$ during the first phase and $\Delta$ during the second phase.

If we denote by
\BE
\sin^2\alpha = \Tr[P_{\tilde{W}}\ket{\psi(K)}\bra{\psi(K)}]
\EE
the winning probability at the end of the first phase according to the first oracle $\tilde{O}$, then the success probability according to the second oracle $O$ at this point is given by 
\BE
p_K=\sin^2\alpha +\cos^2\alpha \frac{n_+}{n_++n_\ell}\label{eq:p1}
\EE
following \eq{eq:psi_K_1} and \eq{eq:beta} in Appendix \ref{app:symmetries}. Here $n_+=|\mathcal{W}_+|$ with $\mathcal{W}_+=\tilde{\mathcal{L}}\cap W$ denotes the number of items $x$ marked only by the second oracle $O$ as winning and $n_\ell=|\mathcal{L|}$ the number of losing items according to $O$. Thus $p_K$  increases monotonically with $\alpha$ for $0\leq \alpha\leq \pi/2$. 

The angle $\phi$ is also upper bounded by $\alpha$ via (see Appendix \ref{app:symmetries}, \eq{eq:phi})
\BE
\tan \phi \leq \tan \alpha \sqrt{\frac{\tilde{n}(n_++n_\ell)}{(\tilde{n}+n_+)n_\ell}}+\sqrt{\frac{n_+^2}{(\tilde{n}+n_+)n_\ell}}.\label{eq:phi}
\EE
This bound also increases monotonically with $\alpha$ for $0\leq \alpha\leq \pi/2$. As a result, the final success probability is upper bounded by the maximal achievable angles $\alpha$ (defined via the strategy during the first phase) and $\Delta$ (during the second phase) within the range $0\leq \alpha \leq \pi/2$ and $0\leq \phi(\alpha)+\Delta\leq \pi/2$. 

The angles $\alpha$ and $\Delta$ can be upper bound with the help of a generalization of the optimality proof of Grover's algorithm from Zalka \cite{Zalka1999} which can be stated in the following way

\begin{lemma}
Given an oracle $O$ which marks exactly $n$ out of $N$ items as winning, then performing Grover's quantum search algorithm gives the maximal possible average success probability $p_K=\sin^2[(2K+1)\nu]$ for up to $0<K<\pi/(4\nu)-1/2$ with $\sin^2 \nu=n/N$. \label{lem:Grover}
\end{lemma}

The proof of this lemma follows the optimality proof from Zalka for $n=1$ given in \cite{Zalka1999}. We outline the difference in the proof for $n>1$ in Appendix \ref{sec:max_diff}. In general, the angle $2K\nu$ does not only limit the maximal success probability via $p\leq \sin^2[(2k+1)\nu]$ when starting from a random guess, equal to  $p_0=\sin^2 \nu=n/N$, but to $p\leq \sin^2[2k\nu+\phi]$ when starting from any fixed initial success probability $p_0=\sin^2\phi$, as we also outline in Appendix \ref{sec:max_diff}.

 As a consequence, the maximal angle $\alpha$ is bounded by $\alpha\leq (2K+1)\tilde{\nu}$ with $\sin^2\tilde{\nu}=\tilde{n}/N$ which follows directly from \lem{lem:Grover} provided $(2k+1)\tilde{\nu}\leq\pi/2$.
And the winning probability of $p_s$ is limited by $\sin^2(\phi+\Delta)$ with $\Delta<2K\nu$ provided $2k\nu+\phi\leq \pi/2$. 


\subsection{Grover search is optimal for monotonically increasing winning spaces\label{sec:Grover}}

\begin{figure}
\begin{center}
\includegraphics[width=0.4\textwidth]{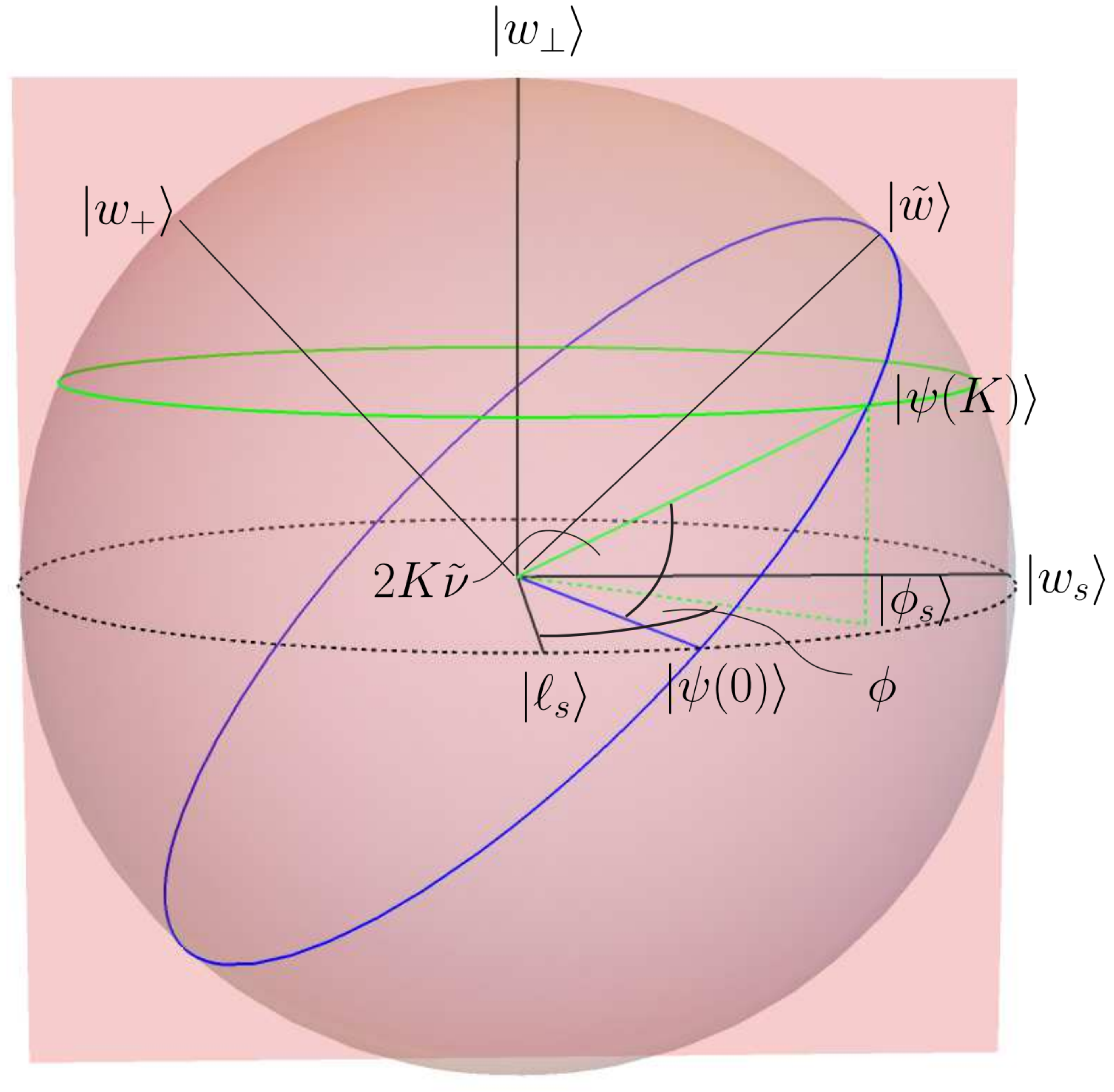}
\caption{Visualization of the time evolution of $\ket{\psi(0)}$ under Grover iterations with changing oracles $O_k=\tilde{O}$ for $1\leq k\leq K$ and $O_k=O$ for $K+1\leq k\leq J$ and $\tilde{\mathcal{W}}\subseteq \mathcal{W}$. The winning space (red plane) of $O$ is spanned by $\lbrace\ket{w_s},\ket{w_\perp}\rbrace$. 
The equal superposition state $\ket{\psi(0)}$ is rotated along the blue circle by an angle $2K\tilde{\nu}$ during the first phase leading to the state $\ket{\psi(K)}$. Consecutively, this state is rotated along the green circle changing only its component $\ket{\phi_s}$ but not $\ket{w_\perp}$.}\label{fig:grover}
\end{center}
\end{figure}

In this section we determine the (average) success probability $p_{K+J}$ for the here defined changing oracle problem obtained via a generalized Grover algorithm and show that it saturates the in \Sec{sec:upper_bound} derived bound. Grover's algorithm starts in a equal superposition state given by
\begin{eqnarray}
\ket{\psi(0)}&=&\frac{1}{\sqrt{N}}\sum\limits_{x=1}^N \ket{x}\\
&=&\sin\tilde{v}\;\ket{\tilde{w}}+\cos\tilde{v}\;\ket{\tilde{\ell}}
\end{eqnarray}
with 
\begin{eqnarray}
\ket{\tilde{w}}&=&\frac{1}{\sqrt{\tilde{n}}}\sum\limits_{x\in \tilde{\mathcal{W}}} \ket{x}\\
\ket{\tilde{\ell}}&=&\frac{1}{\sqrt{N-\tilde{n}}}\sum\limits_{x\in \tilde{\mathcal{L}}} \ket{x}.
\end{eqnarray}
All unitaries $U_k$ for $1\leq k\leq K+J$ are given by
\BE
U_k=\mathds{1}-\ket{\psi(0)}\bra{\psi(0)}.\label{eq:U_K}
\EE
The time evolution during the first phase with oracle $\tilde{O}$ leads to a rotation of $\ket{\psi(0)}$ by an angle $2K\tilde{\nu}$ in the plane spanned by $\ket{\tilde{w}}$ and $\ket{\psi(0)}$ as depicted in \fig{fig:grover}.
The state at the end of the first phase is given by
\BE
\ket{\psi(K)}=\sin[(2K+1)\tilde{\nu}]\ket{\tilde{w}}+\cos[(2K+1)\tilde{\nu}]\ket{\tilde{\ell}}
\EE
and thus saturates the upper limit $\alpha=(2K+1)\tilde{\nu}$ leading to a maximal $p_K$ and $\alpha$. To describe the time evolution during the second phase, we perform a basis transformation into the new basis \begin{eqnarray}
\ket{\ell_s}&=&\frac{1}{\sqrt{n_\ell}}\sum\limits_{x\in \mathcal{L}}\ket{x}\\
\ket{w_s}&=&\frac{1}{\sqrt{n}}\sum\limits_{x\in \mathcal{W}}\ket{x}\\
\ket{w_\perp}&=&\sqrt{\frac{n_+}{\tilde{n}+n_+}}\sum\limits_{x\in \tilde{\mathcal{W}}} \ket{x}-\sqrt{\frac{n}{\tilde{n}+n_+}}\sum\limits_{x\in  \mathcal{W}_+} \ket{x}
\end{eqnarray}
with $\mathcal{W}_+=\tilde{\mathcal{L}}\cap \mathcal{W}$ and $n_+=|W_+|$.
The states $\ket{w_s}$ and $\ket{\ell_s}$ are symmetric under permutations permuting only winning states with winning states and losing states with losing states similar to the symmetry properties of averaged strategies discussed in Appendix \ref{app:symmetries}. The state $\ket{\psi(K)}$ is given in this new basis by
\BE
\ket{\psi(k)}=\cos \varepsilon \Big(\sin \phi \ket{w_s}+\cos \ket{\phi} \ket{\ell_s}\Big)+\sin \phi \ket{w_\perp}
\EE
with the angle $\phi$ defined via
\begin{eqnarray}
\tan\phi &=& \frac{\braket{w_s}{\psi(K)}}{\braket{\ell_s}{\psi(k)}} \\
&=&\tan \alpha \sqrt{\frac{\tilde{n}(n_++n_\ell)}{(\tilde{n}+n_+)n_\ell}}+\sqrt{\frac{n_+^2}{(\tilde{n}+n_+)n_\ell}}
\end{eqnarray}
saturating \eq{eq:phi}. The angle $\varepsilon$ is given by
\BE
\sin \varepsilon =\sqrt{\frac{n_+}{n_++\tilde{n}}}\left[\sin(\alpha)-\sqrt{\frac{\tilde{n}}{n_++n_\ell}}\cos(\alpha)\right].
\EE

The time evolution during the second phase, given by oracle $O$ and $U_k$ as given in \eq{eq:U_K}, leads to a rotation of $\ket{\psi(K)}$ by an angle $2J\nu$ in a plane parallel to the one spanned by $\ket{\psi(0)}$ and $\ket{w}$ as depicted in \fig{fig:grover}. As a consequence, the final state is given by
\BE
\begin{split}
\ket{\psi(K+J)}=&\cos \varepsilon\Big[\sin(\phi+2J\nu)\ket{w_s}+\cos(\phi+2J\nu)\ket{\ell_s}\Big]\\
&+(-1)^J\sin \varepsilon \ket{w_\perp}
\end{split}
\EE
leading to the maximal possible angle $\Delta=2J\nu$ and maximal $p_\perp=\sin^2\varepsilon$ and thus to the maximal possible (average) success probability $p_{K+J}$.

As a result, performing consecutive Grover iterations in the first and second phase with in total $K+J$ oracle queries leads to the maximal possible average success probability $p_{K+J}$ provided $\alpha=(2K+1)\tilde{\nu}\leq \pi/2$ and $\phi(\alpha)+2J\nu\leq \pi/2$.

If more queries are available such that $(2K+1)\tilde{\nu}> \pi/2$ or $\phi+2J\nu> \pi/2$, then it is possible to over-rotate the state $\ket{\psi}$ such that applying $\tilde{O}$ or $O$ less often or performing another algorithm like e.g. fixed-point search \cite{Yoder2014} leads to a higher success probability.  

In general, the change of $\ket{\psi(k)}$ which can be created with a single oracle  query $O$ ($\tilde{O}$) is limited by $|\braket{\psi(k+1)}{\psi(k)}|\geq \cos2\nu$ ($\cos 2\tilde{\nu}$). The maximal possible difference between $\ket{\psi(0)}$ and $\ket{\psi(K+J)}$ achievable under this constrains would require that all states $\ket{\psi(k)}$ ly within a single plane (see discussion in \cite{Zalka1999}). However, changing the oracle in Grover's algorithm leads to a change or tilt of the rotation plane/ axis as visualized in \fig{fig:grover}. Nevertheless, performing Grover iterations is the optimal strategy as we have proven.
 In addition, changing the oracle creates a component $\ket{\phi_\perp}$ which stays invariant under consecutive Grover iterations with the new oracle. Luckily, this component contains only winning items such that it does not prevent us from further increasing the success probability with Grover iterations if $\mathcal{\tilde{W}}\subseteq\mathcal{W}$. As a consequence, the optimality of Grover's algorithm in the case of a changing oracle might be not surprising but is also not obvious. Especially because performing Grover's algorithm with the maximal number of available  oracle queries is not necessarily optimal if $\tilde{\mathcal{W}}$ and $\mathcal{W}$ only share a large overlap but $\tilde{\mathcal{W}}\not\subseteq\mathcal{W}$.

%

\subsection{Grover iterations for $\tilde{\mathcal{W}}\not\subseteq\mathcal{W}$\label{sec:Grover2}}

In the following, we investigate the performance of Grover's algorithm  if $\tilde{\mathcal{W}}$ and $\mathcal{W} $ share a large overlap (see Appendix \ref{app:large overlap}) but $\tilde{\mathcal{W}}\not\subseteq\mathcal{W}$. We will show that performing the maximal number $K$ of oracle queries during the first phase is not always optimal depending on the number of available queries $J$ in the second phase.  

If $\tilde{\mathcal{W}}\not\subset\mathcal{W}$ then the perpendicular component $\ket{\phi_\perp}$, \eq{eq:phi_perp} also includes a losing component $\ket{\ell_\perp}$ such that the state $\ket{\psi(K)}$ can be written via
\begin{eqnarray}
\ket{\psi(K)}&=& \cos \varepsilon \ket{\phi_s}+ \sin \varepsilon \ket{\phi_\perp}\\
\ket{\phi_s}&=&\sin \phi \ket{w_s}+\cos \phi \ket{\ell_s}\\
\ket{\phi_\perp}&=&\sin \chi \ket{w_\perp}+\cos \chi \ket{\ell_\perp}.
\end{eqnarray}
Applying Grover iterations with unitaries $U_k$ as defined in \eq{eq:U_K} does not change the success probability of the component $\ket{\phi_\perp}$. It only changed the success probability of the component $\ket{\phi_s}$ leading to
\begin{eqnarray}
p_{K+J}&=&\cos^2\varepsilon \sin^2(\phi+\Delta)+\sin^2\varepsilon \sin^2\chi\\
&\leq&1-\sin^2\varepsilon\cos^2\chi
\end{eqnarray} 
with $\Delta=2J\nu$. As a consequence, the success probability at the end of second phase is limited by $1-|\braket{\ell_\perp}{\psi(K)}|^2$ and thus by the weight of the orthogonal losing component created during the first phase.

In this case, the success probability $p_{K+J}$
is still monotonically increasing with $\Delta$. Therefore, performing the maximal possible number ($J$) of Grover iterations during the second phase is still a good idea provided $\phi+2J\nu\leq \pi/2$. However, performing the maximal number ($K$) of Grover iterations during the first phase  is not optimal if it leads to phases $\phi=\phi(K)$ and $\chi=\chi(K)$ such that
\BE
\sin^2 \chi(K)< \sin^2[2J\nu+\phi(K)].
\EE 
In this situations, performing less Grover iterations $K'<K$ during the first phase can lead to a higher final success probability $p_{K'+J}>p_{K+J}$.
In general, it is optimal to perform the maximal number $K$ of Grover iterations during the first phase if $J=0$ (provided $(2K+1)\tilde{\nu}<\pi/2$). However, less and less effective queries to the first oracle $\tilde{O}$ should be used the more queries to the second oracle are available as demonstrated in \fig{fig:bad_grover}. 

\begin{figure}
\begin{center}
\includegraphics[width=0.4\textwidth]{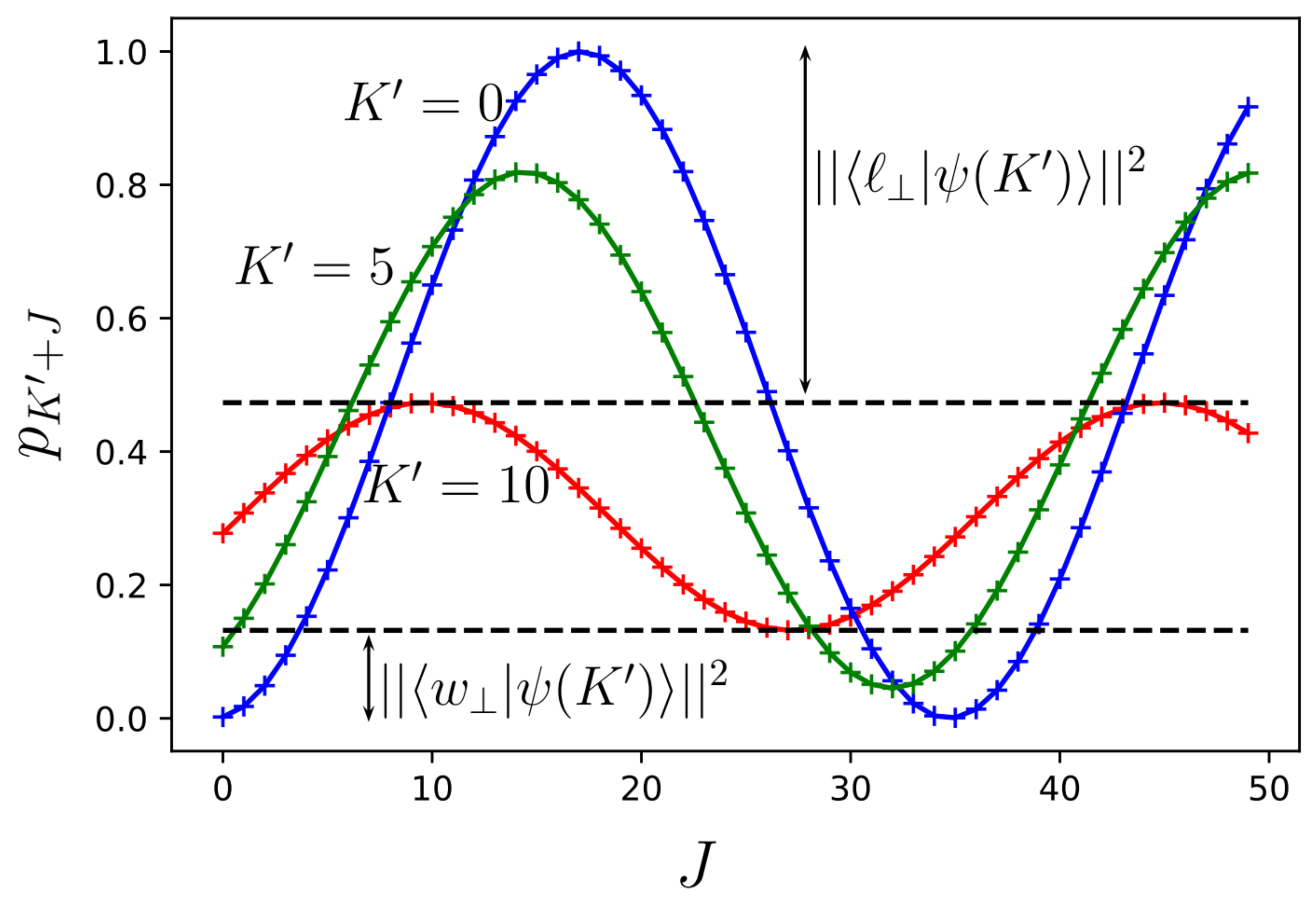}
\caption{Comparison of the success probabilities $p_{K'+J}$ for different numbers $K',J$ of Grover iterations during the first and second phase with $K'=0$ (blue), $K'=5$ (green) and $K'=10$ (red)  for $n_\ell=5000, \tilde{n}=15, n=10,n_+=5$ and thus $n_-=10=|\mathcal{\tilde{W}}\cap \mathcal{L}|$. The success probability for $J=0$ increases with $K'$. However, the maximal possible probability max $ (p_{K'+J})$ maximized overall $J$ decreases with $K'$ such that different $K'(J)$ are optimal for different $J$. }\label{fig:bad_grover}
\end{center}
\end{figure}


\section{Conclusion\label{sec:Conclusion}}

Research in quantum enhanced reinforcement learning has motivated quantum computation scenarios involving two systems, the agent and its environment, with restricted access to each other. In special cases, the interaction of the agent with its environment can be reduced to unitary oracle queries. However, general settings do not allow such a treatment due to memory effects induced by the environment.

In this paper, we generalized  the basic case, where the environment acts effectively as a single fixed oracle, to settings where the oracle changes in time. 
This was  motivated by standard grid-world type problems,  where the number of consecutive actions within a single epoch can grow or shrink. 
We have demonstrated that the search for a winning action sequence of increasing  length can be described as a search in a data based with fixed sequence length (equal to the maximal sequence length) but changing oracle leading to an increase of the winning space. 
We analyzed this setting and identified Grover-type amplitude amplification as optimal strategy for monotonically increasing winning spaces. 

However, continuing coherent Grover iterations when the target space decreases will partially trap the resulting state within the losing subspace. As a consequence, the winning probability will be limited, with a limit clearly below unity,  if we continue with Grover iterations after the oracle has changed.

It is easy to conceive a cascade of ever more general problems. For example, in slightly more general settings the agent might be allowed to chose if and when to change the effective oracle. In this way, the agent might combine breadth-first and depth-first search in a single coherent search for RL. Often, shorter winning action sequences are preferred but longer winning action sequences are more likely. Increasing the sequence length during a coherent quantum search will amplify the probability for shorter winning sequences more than for longer sequences. Combing different oracles, corresponding to different sequence lenght, within a single Grover search might therefore  help to balanced the tradeoff between the desire for short winning sequences on the one side and high winning probabilities on the other.

The goal in RL is in general to minimize a given cost function instead of maximizes solely the success probability. In general, performing consecutive Grover iterations  can be also used to minimize the average number of oracle queries necessary until a winning item is found.  An optimal algorithm will depend on the exact cost function we want to minimize. For example, the search algorithm described in \cite{Boyer1998} is only optimal in terms of oracle queries. However, the number of elementary qubit gates necessary to perform a Grover search can be reduced by using a recursive Grover search \cite{Wolf2015} which separates the database into several subgroups. In RL, queries to different oracles  might be connected to different cost. In such setting, an optimal algorithm might  use different oracles  in a recursive way for a quantum search.

Finally, possibly the most interesting extensions  would avoid reductions of environments to unitary oracles, and identify new schemes to obtain improvements in settings which may be more applicable in real-world RL settings. We leave these more general considerartionions for follow-up investigations.


\section*{Acknowledgment}

The authors thanks H.J. Briegel, F. Flamini, S. Jerbi, D. Orsucci and L. Trenkwalder for fruitful discussions. SW acknowledges support from the Austrian Science Fund (FWF) through the SFB BeyondC. AH  acknowledges support from the Austrian Science Fund (FWF) through the project P 30937-N27. This work was also supported by the Dutch Research Council (NWO/OCW), as part of the Quantum Software Consortium programme (project number 024.003.037).

\pagebreak
\begin{appendix}
\begin{widetext}
\section{Large overlap of $\tilde{\mathcal{W}}$ and $\mathcal{W}$\label{app:large overlap}}
We say that the winning spaces $\tilde{\mathcal{W}}$ and $\mathcal{W}$ have a large overlap if increasing the probability $\tilde{p}$ for $x\in \tilde{\mathcal{W}}$ uniformly also increases the probability $p$ to find $x\in {\mathcal{W}}$.

In general, optimal search strategies can be always constructed in such a way that the probability for all winning states $p(x|f(x)=1)$ are equal as outlined in Appendix \ref{app:averaged_strategy}. The same holds for losing states $\ket{x}$ with $f(x)=0$. Let $n_a=|\tilde{\mathcal{W}}\cap \mathcal{W}|$ ($n_\ell$) be the number of states which are marked as winning (losing) by both oracles and $n_-=|\tilde{\mathcal{W}}\cap \mathcal{L}|$ ($n_+$) the number of states which win only according to the first (second) oracle. Thus, the total number of items is given by $N=n_a+n_\ell+n_-+n_+$. We denote the probabilities to find any state which always wins, always loses, wins only in the first phase and wins only during the second phase by $p_a,p_{\ell},p_-,p_+$. Increasing the initial probability $\tilde{p}=p_a+p_-=(n_a+n_-)/N$ during the first phase in a symmetric way as outlined in Appendix \ref{app:averaged_strategy} by a factor $\alpha$ leads to
\BE
p_a\rightarrow \alpha p_a,\quad p_-\rightarrow \alpha p_-,\quad p_{\ell}\rightarrow \beta p_{\ell},\quad p_+\rightarrow \beta p_+
\EE
with
\BE
\beta= \frac{N-\alpha (n_a+n_-)}{n_\ell+n_+}
\EE
due to normalization.
This leads to a change of $p$ given by
\BE
p=p_a+p_+\rightarrow \frac{n_+}{n_{\ell}+n_+}+\alpha \frac{n_an_{\ell}-n_+n_-}{n_{\ell}+n_+}.
\EE
As a result, we can increase  $p$ by increasing $\tilde{p}$ in a symmetric way whenever
\BE
n_an_{\ell}>n_+n_-\label{eq:large_overlap1}.
\EE
As a result, we say $\tilde{W}$ and $W$ share a large overlap if they fulfill \eq{eq:large_overlap1}.

\section{Averaged search strategies\label{app:averaged_strategy}}
In the following we consider search problems defined via some set of $N$ orthonormal  states $\lbrace\ket{n}_A\rbrace$ forming the basis of the Hilbert space $\mathcal{H}_A$ which can be separated into two subsets $\mathcal{H}_A=\mathcal{W}\cup\mathcal{L}$,  the set of winning states $\mathcal{W}$ and the set of losing states $\mathcal{L}$ with $\mathcal{W}\cap \mathcal{L}=\emptyset$. Information about winning states can be obtained by querying phase-flip oracles
\BE
O_k=P_{\mathcal{W}_k}-P_{\mathcal{L}_k}.
\EE
where $P_{\mathcal{W}_k}$ and $P_{\mathcal{L}_k}$  denote projectors on some subspaces $\mathcal{W}_k$ and $\mathcal{L}_k$ forming together again the complete Hilbert space $\mathcal{H}_A=\mathcal{W}_k\cup\mathcal{L}_k$ with $\mathcal{W}_k\cap \mathcal{L}_k=\emptyset$. For standard search problems we have $\mathcal{W}_k=\mathcal{W}\;\forall k$ and $\mathcal{L}_k=\mathcal{L}\;\forall k$. However, for more general search problems such as the here considered changing oracle problem, the subspaces $\mathcal{W}_k$ and $\mathcal{L}_k$ might differ from query to query.

Our goal is to find any state $\ket{n}_A\in\mathcal{W}$ with the help of maximally $K$ oracle queries. All possible search strategies can be represented via unitary operations and pure initial states since it is possible to purify any search strategy by going to a larger Hilbert space  $\mathcal{H}_{AB}=\mathcal{H}_{A}\otimes\mathcal{H}_B$ and defining the generalize operators
\BE
O_{AB}=O_{A}\otimes\mathds{1}_B, \;P_{\mathcal{W},AB}=P_{\mathcal{W},A}\otimes\mathds{1_B},\; P_{\mathcal{L},AB}=P_{\mathcal{L},A}\otimes\mathds{1_B}.\label{eq:purification}
\EE
To avoid a notation with over boarding indices, we skip the labels indicating the  different subspace the operators/unitaries are working on if they are not crucial. Operators with a subspace index, such as e.g. $\sigma_A$ acting on a state from a larger Hilbert space, e.g. $\ket{\psi}_{AB}$ are meant as short forms of the generalized operators defined similar to \eq{eq:purification}.     

Any search strategy $T$ to find a state $\ket{n}\in \mathcal{W}$ can be described via $T=(\lbrace{U_k}\rbrace,\ket{\psi(0)}_{AB})$ with a pure initial state $\ket{\psi(0)}_{AB}$ and unitaries $\lbrace{U_k}\rbrace$ acting on the combined Hilbert space $\mathcal{H}_{AB}$  leading  after $K$ oracle queries to the final state
\BE
\ket{\psi(K)}_{AB}=U_KO_K\cdots U_2O_2U_1 O_1 \ket{\psi(0)}_{AB} 
\EE
and a consecutive projective measurement.
Without loss of generality, we apply first an oracle query since any unitary $U_0$ applied before can be subsumed into the initial state. The probability $p_T$ to identify a winning state correctly for a given strategy $T$ and set of oracles $\lbrace O_k\rbrace$ is then given by
\BE
p_T=\Tr\Big(P_\mathcal{W}\ket{\psi(K)}\bra{\psi(K)}\Big).
\EE

Let $\sigma$ denote a permutation operator acting on $\mathcal{H}_A$ and $\Sigma_A$ denoting the group of operators of all possible permutations.  
The average winning probability $\bar{p}_T$ of the strategy $T$ is defined via
\BE
\bar{p}_T=\frac{1}{N!}\sum\limits_{\sigma \in \Sigma_A}p_T(\sigma) =\frac{1}{N!}\sum\limits_{\sigma \in \Sigma_A}\Tr\Big[\sigma P_\mathcal{W}\sigma^\dagger \ket{\psi(K,\sigma)}\bra{\psi(K,\sigma)}\Big]
\EE
with
\BE
\ket{\psi(K,\sigma)}_{AB}=U_K\sigma O_K\sigma^\dagger\cdots U_1 \sigma O_1\sigma^\dagger \ket{\psi(0)}_{AB} 
\EE
being the resulting state if we substitute every oracle $O_k$ by $\sigma O_k\sigma^\dagger$.

For every search strategy $T_{AB}=(\lbrace{U_{k,AB}}\rbrace,\ket{\psi(0)}_{AB})$ we can  define an averaged strategy $\bar{T}_{ABC}$ via

\begin{definition}
The averaged strategy $\bar{T}_{ABC}=(\lbrace{\bar{U}_{k,ABC}}\rbrace,\ket{\bar{\psi}(0)}_{ABC})$ of  the strategy $T_{AB}=(\lbrace{U_{k,AB}}\rbrace,\ket{\psi(0)}_{AB})$
is defined via the averaged initial state
\BE
\ket{\bar{\psi}(0)}_{ABC}=\frac{1}{\sqrt{N!}}\sum\limits_{\sigma_\gamma\in \Sigma_A} \sigma_\gamma^\dagger\ket{\psi(0)}_{AB}\ket{\gamma}_C.
\EE
and average unitaries
\BE
\bar{U}_{k,ABC}=\sum\limits_{\sigma_\gamma \in \Sigma_A} \sigma_\gamma^\dagger U_{k,AB} \sigma_\gamma\otimes \ket{\gamma}_C\bra{\gamma}.
\EE
\end{definition}
Here, the states $\lbrace \ket{\gamma}_C\rbrace$ are given by an arbitrary orthonormal basis of a Hilbert space $\mathcal{H}_C$ with dimension $d_C=N!$ acting as labels for the applied permutation operator $\sigma_\gamma$ acting on $\mathcal{H}_A$. 

The averaged strategy $\bar{T}$ has the following properties:

\begin{lemma}
The success probability $p_{\bar{T}}(\sigma)$ of the averaged strategy $\bar{T}$ is equal to the average success probability $\bar{p}_T$ of the strategy $T$ for every permutation $\sigma\in\Sigma_A$. 
\label{lemma:sym_strat}
\end{lemma}

\begin{proof}
The success probability $p_{\bar{T}}(\sigma)$ is given by
\BE
p_{\bar{T}}(\sigma)=\Tr_{ABC}\left[\sigma P_{\mathcal{W},ABC}\sigma^\dagger \ket{\bar{\psi}(K,\sigma)}\bra{\bar{\psi}(K,\sigma)}\right].
\EE
The state $\sigma^\dagger\ket{\bar{\psi}(K,\sigma)}_{ABC}$ for $\sigma \in \Sigma_A$ is given by
\begin{eqnarray}
\sigma^\dagger\ket{\bar{\psi}(K,\sigma)}_{ABC}&=&\frac{1}{\sqrt{N!}}\sum\limits_{\sigma_\gamma\in \Sigma_A}{\sigma}^\dagger{\sigma}_\gamma^\dagger U_K{\sigma}_\gamma\sigma O_K{\sigma}^\dagger\cdots {\sigma}_\gamma^\dagger U_1 {\sigma}_\gamma \sigma O_1{\sigma}^\dagger   {\sigma}_\gamma^\dagger\ket{\psi(0)}_{AB}\ket{\gamma}_C\\
&=&\frac{1}{\sqrt{N!}}\sum\limits_{\tilde{\sigma}_\gamma\in \Sigma_A}\tilde{\sigma}_\gamma^\dagger U_K{\tilde{\sigma}_\gamma}O_K\cdots {\tilde{\sigma}_\gamma}^\dagger U_1 \tilde{\sigma}_\gamma O_1   {\tilde{\sigma}_\gamma}^\dagger\ket{\psi(0)}_{AB}\ket{\gamma}_C
\end{eqnarray}
where we used
\BE
\sigma\sum\limits_{\sigma_\gamma\in\Sigma_A}\sigma_\gamma =\sum\limits_{\tilde{\sigma}_\gamma\in\Sigma_A}\tilde{\sigma}_\gamma \quad \forall \sigma\in\Sigma_A
\EE
because $\Sigma_A$ is a symmetric group. As a consequence, the application of the permutation $\sigma^\dagger$ on $\ket{\bar{\psi}(K)}$ is equivalent to a relabeling of the permutations $\sigma_\gamma$ such that we now apply the permutation $\tilde{\sigma}_\gamma^\dagger=\sigma^\dagger\sigma_\gamma^\dagger $ instead of $\sigma_\gamma^\dagger$ if subsystem $C$ is in state $\ket{\gamma}_C$. However, these labels have been arbitrary and therefore we find for the success probability 
\begin{eqnarray}
p_{\bar{T}}(\sigma)&=&\Tr_{ABC}\left[ P_{\mathcal{W},ABC}\sigma^\dagger \ket{\bar{\psi}(K,\sigma)}\bra{\bar{\psi}(K,\sigma)}\sigma\right]\\
&=&\frac{1}{N!}\sum\limits_{\tilde{\sigma} \in \Sigma_{\mathcal{H}_A}}\Tr_{AB}\left[\tilde{\sigma}P_{\mathcal{W},AB}{\tilde{\sigma}}^\dagger \ket{\psi(K),\tilde{\sigma}}_{AB}\bra{\psi(K),\tilde{\sigma}}_{AB} 
\right]\\
&=&\bar{p}_T
\end{eqnarray}
\qed
\end{proof}

The relabeling can be formalized in the following way. We define the index $\tilde{\gamma}$ via $\sigma \sigma_\gamma= \sigma_{\tilde{\gamma}}$. Then, we can define the permutation $\pi(\sigma)$ acting on $\mathcal{H}_C$ via
\BE
\pi(\sigma)\ket{\gamma}_C=\ket{\tilde{\gamma}}_C
\EE
which then leads to the following lemma:

\begin{lemma}
The averaged strategy $\bar{T}$ is permutation invariant under joined permutations $\sigma\otimes \pi(\sigma)$  $\forall \sigma\in \Sigma_{A}$  such that
\begin{eqnarray}
[\bar{U}_k,\sigma\otimes \pi(\sigma)]&=&0\\
\sigma\otimes \pi(\sigma)\ket{\bar{\psi}(0)}&=&\ket{\bar{\psi}(0)}.
\end{eqnarray}
\label{lemma:perm_inv}
\end{lemma}

\begin{proof}
For the symmetric initial state $\ket{\bar{\psi}(0)}$, we find
\begin{eqnarray}
\sigma\otimes\pi(\sigma)\ket{\bar{\psi}(0)}_{ABC} &=& \sigma\otimes \pi(\sigma) \frac{1}{\sqrt{N!}}\sum\limits_{\sigma_\gamma\in\Sigma_A} \sigma_\gamma\ket{\psi(0)}_{AB}\ket{\gamma}_C\\
&=& \frac{1}{\sqrt{N!}}\sum\limits_{\sigma_\gamma \in \Sigma_A} \sigma\sigma_\gamma\ket{\psi(0)}_{AB}\pi(\sigma)\ket{\gamma}_C\\
&=&\frac{1}{\sqrt{N!}}\sum\limits_{\sigma_{\tilde{\gamma}} \in \Sigma_A} \sigma_{\tilde{\gamma}}\ket{\psi(0)}_{AB}\ket{\tilde{\gamma}}_C=\ket{\bar{\psi}(0)}_{ABC}.
\end{eqnarray}
For the symmetric unitaries $\bar{U}_k$ we find
\begin{eqnarray}
\sigma\otimes \pi(\sigma)\bar{U}_k{\sigma}^\dagger\otimes \pi^\dagger(\sigma)&=&
\sigma\otimes \pi(\sigma)\left(\sum\limits_{\sigma_\gamma \in \Sigma_A} \sigma_\gamma U_{k,AB} \sigma_\gamma^\dagger \otimes \ket{\gamma}_C\bra{\gamma}\right){\sigma}^\dagger\otimes \pi^\dagger(\sigma)\\
&=&\sum\limits_{\sigma_{\tilde{\gamma}} \in \Sigma_A} \sigma_{\tilde{\gamma}} U_{k,AB} \sigma_{\tilde{\gamma}}^\dagger \otimes \ket{\tilde{\gamma}}_C\bra{\tilde{\gamma}}\\
&=&\bar{U}_k.
\end{eqnarray}
$[\bar{U}_k,\sigma\otimes \pi(\sigma)]=0$ follows immediately since permutation operators are unitary.
\qed
\end{proof}

As a consequence of \lem{lemma:sym_strat} and \lem{lemma:perm_inv}, we can limit the search for the best strategy $T$, optimizing $\bar{p}_T$, to averaged strategies $\bar{T}$ which also optimize the worst case probability $\underset{\sigma}{\text{min }}p_T(\sigma)$ and leads to certain symmetries as outlined in Appendix \ref{app:symmetries}. 


\section{Symmetry investigations for the changing oracle problem\label{app:symmetries}}
In the following, we consider a search problem, where the oracle $O_k$ changes at a certain time step. Thus we can separate the search into two phases. The first phase contains $K$ oracle queries to oracle $\tilde{O}=O_k$ for $1\leq k\leq K$ with winning space $\tilde{\mathcal{W}}$ and losing space $\tilde{\mathcal{L}}$. Then, the oracle changes to $O=O_k$ for $K<k\leq J$ with the new winning space ${\mathcal{W}}$ and losing space ${\mathcal{L}}$ and the search is continued by another $J$ queries to ${O}$. In addition, we restrict the problem to monotonically increasing winning spaces that is the winning space $\tilde{\mathcal{W}}$ of the first phase is a subset $\tilde{\mathcal{W}}\subseteq \mathcal{W}$ of the winning space $\mathcal{W}$ of the second oracle $O$. This automatically leads to $\tilde{\mathcal{L}}\supseteq \mathcal{L}$.

In the following, we investigate the appearing symmetries occurring during the first and second phase when applying averaged search strategies $\bar{T}$ to this problem. Since we only consider averaged strategies and thus averaged unitaries $\bar{U}_k$ and states $\ket{\bar{\psi}(k)}$, we omit the bar  on all states and unitaries in this section to simplify the notation.

In the following, we investigate the symmetry properties of the states
\begin{eqnarray}
\ket{\psi(K)}&=&U_K O_K \cdots U_1 O_1\ket{\psi(0)} \\
\ket{\psi(K+J)}&=&U_{K+J}O_{K+J}\cdots U_{K+1}O_{K+1}\ket{\psi(K)}
\end{eqnarray}
at the end of the first and the second phase. This will allow us to determine an upper bound for the average success probability ${p}$.

We define the set of permutations (operators) $\Sigma_{\tilde{O}}=\Sigma_{\tilde{\mathcal{W}}}\cup \Sigma_{\tilde{\mathcal{L}}}$ as the complete set of permutations operators which leave the winning space $\tilde{\mathcal{W}}$ and losing space $\tilde{\mathcal{L}}$ invariant. As a consequence, we find $[\tilde{O},\sigma]=0\,\forall \sigma\in \Sigma_{\tilde{O}}$. 
The initial state $\ket{\psi(0)}$ and all unitaries $U_k$ and $\tilde{O}_k$ during the first phase are permutation invariant under $\sigma\otimes \pi(\sigma)$ $ \forall \sigma\in \Sigma_{\tilde{O}}$ since
$\Sigma_{\tilde{O}}\subseteq \Sigma_{\mathcal{H}_A}$. Thus, the state $\ket{\psi(K)}$ 
at the end of the first phase is also permutation invariant under $\sigma\otimes \pi(\sigma)$ $ \forall \sigma\in \Sigma_{\tilde{O}}$. 

To determine the symmetry properties of $\ket{\psi(K+J)}$ we need to investigate how the winning and losing components of $\ket{\psi(K)}$ changes when we change the oracle.
We define the normalized winning $\ket{\tilde{w}}$ and losing component $\ket{\tilde{\ell}}$ of $\ket{\psi(K)}$ via
\begin{eqnarray}
\cos\alpha\ket{\tilde{\ell}}&=&P_{\tilde{\mathcal{L}}}\ket{\psi(K)}\\
\sin\alpha\ket{\tilde{w}}&=&P_{\tilde{\mathcal{W}}}\ket{\psi(K)}
\end{eqnarray}
with $\cos \alpha = |P_{\tilde{\mathcal{L}}}\ket{\psi(K)}|$. As a consequence, $\ket{\psi(K)}$ can be decomposed via
\BE
\ket{\psi(K)}=\cos \alpha \ket{\tilde{\ell}}+\sin{\alpha}\ket{\tilde{w}}.\label{eq:psi_1}
\EE
The components $\ket{\tilde{w}}$ and $\ket{\tilde{\ell}}$ are permutation invariant under $\sigma\otimes \pi(\sigma)$ $ \forall \sigma\in \Sigma_{\tilde{O}}$ because the projectors $P_{\tilde{\mathcal{W}}}$ and $P_{\tilde{\mathcal{L}}}$ as well as $\ket{\psi(K)}$ are permutation invariant. 

Let us now investigate the winning and losing components at the beginning of the second phase. The initial state of the second phase is given by $\ket{\psi(K)}$. Its component $\ket{\tilde{w}}$ is also a winning component according to the second oracle $O$ such that $P_{\mathcal{W}}\ket{\tilde{w}}=\ket{\tilde{w}}$. However, $\ket{\tilde{\ell}}$ contains both winning and losing components
\begin{eqnarray}
\cos \beta\ket{\ell}&=&P_{{\mathcal{L}}}\ket{\tilde{\ell}}\\
\sin \beta\ket{w_+}&=&P_{{\mathcal{W}}}\ket{\tilde{\ell}}
\end{eqnarray}
with $\cos \beta= |P_{{\mathcal{L}}}\ket{\tilde{\ell}}|$. Note, that
 $\ket{w_+} \in \mathcal{W}_+=\tilde{\mathcal{L}}\cap \mathcal{W}$ and thus $\ket{w_+}\perp \ket{\tilde{w}}$.  Therefore, we can divide the state $\ket{\psi(K)}$ into three orthogonal components via
\BE
\ket{\psi(K)}=\sin\alpha \ket{\tilde{w}}+\cos\alpha \Big(\sin \beta \ket{w_+}+\cos \beta \ket{\ell}\Big).\label{eq:psi_K_1}
\EE
The angle $\beta$ is given by
\BE
\sin \beta = \sqrt{\frac{n_+}{n_++n_\ell}}\label{eq:beta}
\EE
where $n_+$ denotes the dimension of $\mathcal{W}_+$ and $n_\ell$ the dimension of $\mathcal{L}$ (see Appendix \ref{app:beta}).

Let us now invest the symmetries of $\ket{\psi(K)}$ with respect to permutations $\sigma\otimes \pi(\sigma)$ $\forall \sigma\in \Sigma_{O}$ which leave the second oracle $O$ invariant. Let 
$P_\mathcal{S}$ be the projector onto the symmetric subspace which can be written as
\BE
P_\mathcal{S}= \sum\limits_{\ket{s}}\ket{s}\bra{s} \quad \text{with} \quad \sigma\otimes \pi(\sigma)\ket{s}=\ket{s}\quad \forall \sigma\in \Sigma_{O}
\EE
where $\lbrace \ket{s}\rbrace$ forms an orthonormal basis of the symmetric subspace. 
Then, we can define the symmetric component 
\BE
\cos \varepsilon\ket{\phi_s}=P_\mathcal{S}\ket{\psi(K)}\label{eq:phi_sym1}
\EE 
and its complement
\BE
\sin \varepsilon\ket{\phi_\perp}=(\mathds{1}-P_\mathcal{S})\ket{\psi(K)}
\EE
with $\cos \varepsilon=|P_\mathcal{S}\ket{\psi(K)}|$.
The state $\ket{\ell}$ is permutation invariant under $\sigma\otimes \pi(\sigma)$ $ \forall \sigma_A\in \Sigma_{{O}}$ since $\mathcal{L}\subseteq\tilde{\mathcal{L}}$ such that $P_\mathcal{S}\ket{\ell}=\ket{\ell}$. However, the (not normalized) winning component $\sin\alpha \ket{\tilde{w}}+\cos\alpha \sin \beta \ket{w_+}$ is not necessarily permutation invariant under $\sigma\otimes \pi(\sigma)$ $ \forall \sigma\in \Sigma_{{O}}$. As a consequence, there might exist a non-vanishing component $\ket{\phi_\perp}$, however, this component lies within the winning space $\mathcal{W}$ such that
\BE
\sin \varepsilon\ket{\phi_\perp}=(\mathds{1}-P_S)\ket{\psi(K)}=P_\mathcal{W}(\mathds{1}-P_S)\ket{\psi(K)}=\sin \varepsilon\ket{w_\perp}
\EE
The symmetric component $\ket{\phi_S}$ can be decomposed into a winning and a losing component
\begin{eqnarray}
\cos \varepsilon \sin \phi\ket{w_s}&=&P_\mathcal{W}P_\mathcal{S}\ket{\psi(K)}\label{eq:phi_sym2}\\
\cos \varepsilon \cos \phi\ket{\ell_s}&=&P_\mathcal{L}P_\mathcal{S}\ket{\psi(K)}=\cos \varepsilon \cos \phi\ket{\ell}\label{eq:phi_sym3}
\end{eqnarray}
with $\cos\varepsilon\sin \phi=|P_\mathcal{W}P_\mathcal{S}\ket{\psi(K)}|$.
Thus the state $\ket{\psi(K)}$ can be separated into the following three orthogonal components
\BE
\ket{\psi(K)}=\cos{\varepsilon}\Big(\sin{\phi}\ket{w_s}+\cos \phi\ket{\ell}\Big  )+\sin{\varepsilon}\ket{w_\perp}.
\EE
A comparison with \eq{eq:psi_K_1} leads to the following identities
\begin{eqnarray}
\cos \varepsilon \cos\phi &=& \braket{\ell}{\psi(K)}= \cos \alpha \cos \beta \label{eq:cos_phi}\\
\cos \varepsilon \sin \phi&=&\braket{w_s}{\psi(K)}= \sin\alpha \braket{w_s}{\tilde{w}}+\cos\alpha\sin\beta\braket{w_s}{{w}_+}\label{eq:sin_phi}\\
\sin\varepsilon &=&\braket{w_\perp}{\psi(K)}= \sin\alpha \braket{w_\perp}{\tilde{w}}+\cos\alpha\sin\beta\braket{w_\perp}{{w}_+}.
\end{eqnarray}
Note, all appearing scalar products are real due to the definition of $\ket{w_s}$ and $\ket{w_\perp}$ and they are upper bounded via
\begin{eqnarray}
|\braket{w_s}{\tilde{w}}|&\leq& \sqrt{\frac{\tilde{n}}{\tilde{n}+n_+}}\\
|\braket{w_s}{{w}_+}|&\leq& \sqrt{\frac{{n}_+}{{n}_\ell+n_+}}.
\end{eqnarray}
As a consequence, the angle $\phi$ is upper bounded by the angle $\alpha$ via
\BE
\tan \phi \leq \tan \alpha \sqrt{\frac{\tilde{n}(n_++n_\ell)}{(\tilde{n}+n_+)n_\ell}}+\sqrt{\frac{n_+^2}{(\tilde{n}+n_+)n_\ell}}.\label{eq:phi}
\EE

Let us investigate the time evolution during the second phase. We denote with
\BE
V=U_{K+J}O \cdots U_{K+1}O\label{eq:V}
\EE
a unitary which described the complete time evolution during the second phase. The unitary $V$ commutes with the projector $P_\mathcal{S}$ as the following considerations will prove. There exist a joined eigenbasis of $V$ and $\sigma\otimes \pi(\sigma)$ $\forall \sigma\in \Sigma_O$   since $[V,\sigma\otimes \pi(\sigma)]=0$. Let $\lbrace \ket{v_x}\rbrace $ be an eigenbasis of $V$ and wlog we assume that the first $f$ states of this basis form the symmetric subspace such that
\BE
\sigma\otimes \pi(\sigma)\ket{v_x}=\ket{v_x} \quad \forall \sigma \in \Sigma_O \text{ and }1\leq x\leq f.
\EE
As a consequence, we find
\BE
P_\mathcal{S}V=\sum\limits_{x=1}^f \ket{v_x}\bra{v_x}\sum\limits_y \lambda_y \ket{v_y}\bra{v_y}=\sum\limits_{x=1}^f \lambda_x\ket{v_x}\bra{v_x}=VP_\mathcal{S}
\EE
where $\lambda_y$ denote the eigenvalues of $V$. Thus the time evolution of the symmetric component $V\ket{\phi_S}$ stays a symmetric state with
\BE
P_\mathcal{S}V\ket{\phi_S}=VP_\mathcal{S}\ket{\phi_S}=V\ket{\phi_S}
\EE
whereas $V\ket{\phi_\perp}$ stays orthogonal to this subspace since
\BE
P_\mathcal{S}V\ket{\phi_\perp}=VP_\mathcal{S}\ket{\phi_\perp}=0
\EE 
and thus the symmetric part and the orthogonal part do not mix.

The winning probability of $\ket{\psi(K+J)}$ can be decomposed into a symmetric part and a part orthogonal to it via
\begin{eqnarray}
\Tr\left[P_{\mathcal{W}}\ket{\psi(K+J)}\bra{\psi(K+J)}\right]&=&\Tr\left[P_{\mathcal{W}}P_\mathcal{S}\ket{\psi(K+J)}\bra{\psi(K+J)}\right]\nonumber\\
&&+\Tr\left[P_{\mathcal{W}}(\mathds{1}-P_\mathcal{S})\ket{\psi(K+J)}\bra{\psi(K+J)}\right]\\
&=&\cos^2 \varepsilon \Tr\left[P_{\mathcal{W}} V\ket{\phi_s}\bra{\phi_s}V^\dagger\right] \nonumber \\
&&+\sin^2 \varepsilon \Tr\left[P_{\mathcal{W}} V\ket{\phi_\perp}\bra{\phi_\perp}V^\dagger\right]\label{eq:pspp}
\end{eqnarray}
where we used $[P_\mathcal{W},P_\mathcal{S}]=0$ which follows directly from $[P_\mathcal{W},\sigma\otimes \pi(\sigma)]=0$ $\forall \sigma\in \Sigma_O$ and $P_S=P_S^2$.


\section{Determining the angle $\beta$ \label{app:beta}}
In the following, we give a more detailed derivation of \eq{eq:beta} for determining $\beta$ defined via
\BE
\sin^2 \beta = \bra{\tilde{\ell}}P_\mathcal{W}\ket{\tilde{\ell}}.
\EE
Let wlog $\lbrace \ket{j}_A\rbrace$ with $1\leq j\leq n_+ + n_\ell$ be a basis of the losing space $\tilde{\mathcal{L}}_A$. The state $\ket{\tilde{\ell}}_{ABC}$ can then be written as
\BE
\ket{\tilde{\ell}}=\sum\limits_{j=1}^{n_++n_\ell}\xi_j \ket{j}_A\ket{\gamma_j}_{BC}
\EE
with some arbitrary normalized states $\ket{\gamma_j}_{BC}$. The probability for each state $\ket{j}_A$ is given by
\BE
p_j=|| \;\ket{j}_A\bra{j}\otimes \mathds{1}_{BC}\ket{\tilde{\ell}}||^2=|\xi_j|^2.
\EE
However, the state $\ket{\tilde{\ell}}$ is permutation invariant  $\forall \sigma \in \Sigma_{\tilde{\mathcal{L}}}$ such that
\begin{eqnarray}
\sigma\otimes\pi(\sigma)\ket{\tilde{\ell}}&=&\sum\limits_{j=1}^{n_++n_\ell}\xi_j \ket{j'(\sigma)}_A\ket{\gamma'_j}_{BC}=\ket{\tilde{\ell}}
\end{eqnarray}
with $\ket{\gamma'_j}_{BC}=\mathds{1}_B\otimes\pi_C\ket{\gamma_j}_{BC}$. As a consequence, we find for the probabilities 
\BE
p_j=|| \;\Big(\ket{j}\bra{j}\otimes \mathds{1}_{BC}\Big)\; \Big(\sigma_A\otimes\pi_C(\sigma)\Big)\ket{\tilde{\ell}}||^2=|\xi_{j'}|^2\quad \forall\quad 1\leq j,j'\leq n_++n_\ell
\EE
and due to normalization $p_j=1/(n_++n_\ell)$. Since there exist $n_+$ orthonormal states within the subspace $\mathcal{W}_+=\mathcal{W}\cap\tilde{\mathcal{L}}$ we find
\BE
\sin^2\beta = \frac{n_+}{n_++n_\ell}.
\EE
\section{Optimality proof of Grover's algorithm for multiple winning items\label{sec:max_diff}}

The optimality proof of Grover's algorithm for oracles with a single winning item by Zalka \cite{Zalka1999} consist of two parts given by the inequality
\BE
2N-2\sqrt{Np}-2\sqrt{N(N-1)(1-p)}\leq \sum\limits_{y=1}^N|| \ket{\phi_J}-\ket{\phi_J^y}||^2\leq 4N\sin^2(J\psi).\label{eq:zalka}
\EE
Here, $N$ is the number of items, $p$ the success probability to identify the single winning item $y$ correctly, $J$ the maximal number of oracle queries  and the angle $\psi$ is defined via $\sin^2\psi=1/N$. 
The two quantum states $\ket{\phi_j}$ and $\ket{\phi_j^y}$ are defined via
\begin{eqnarray}
\ket{\phi_j}=V^j\ket{\phi}\\
\ket{\phi_j^y}=V^j_y\ket{\phi}
\end{eqnarray}	
where $\ket{\phi}$ is some arbitrary state, $V_y^j$ a unitary of the form \eq{eq:V} based on $j$ queries to the oracle $O_y$  and  $V^j$ is a unitary based on  $j$ queries to an empty oracle. The optimality of Grover's algorithm follows from the proof of both inequalities and the fact that Grover's algorithm saturates both.

We generalize the results from Zalka by going to oracles $O_y$ which mark exactly $n$ items out of $N$ items as winning. In this case, $y$ is now a label for the winning space $\mathcal{W}_y$ and there exist now $D={N \choose n}$ different oracles. The success probability $p$ now denotes the probability to identify any winning item $\ket{z}\in \mathcal{W}_y$ correctly. For a random guess, this probability is given by $\sin^2\nu=n/N$. As a consequence, \eq{eq:zalka} can be generalized to
\BE
2D-2D\sqrt{p\frac{n}{N}}-2D\sqrt{(1-p)\left(1-\frac{n}{N}\right)}\leq \sum\limits_{y=1}^D|| \ket{\phi_J}-\ket{\phi_J^y}||^2\leq 4D\sin^2(J\nu) \label{eq:woelk}
\EE
which we will proof in the following and is equal to \eq{eq:zalka} for $n=1$. Again, Grover's algorithm saturates these bounds.

We start with the right inequality and proof the following lemma

\begin{lemma}
The maximal difference between $\ket{\phi_J}$ and $\ket{\phi_J^y}$ achievable with $J$ oracle queries averaged over all possible oracles with $n$ winning items is given by
\BE
\frac{1}{D}\sum\limits_{y=1}^D|| \ket{\phi_J}-\ket{\phi_J^y}||^2\leq 4\sin^2(J\psi)=2[1-\cos(2J\nu)]\label{eq:max_diff}
\EE
with $\sin^2\nu=n/N$.
\label{lemma:distance}
\end{lemma}

\begin{proof}
This Lemma follows directly from the optimality proof of Grover's algorithm given in \cite{Zalka1999} by generalizing the sum overall possible oracles which mark only one item $y$ to all possible oracles which mark $n$ items.  In the following, we do not reproduce every step from Ref. \cite{Zalka1999} but concentrate only on steps where the generalization from one winning item to several winning items makes a difference.   Following Ref.  \cite{Zalka1999} we find (eq. 22) 
\BE
\frac{1}{D}\sum\limits_{y=1}^D || \ket{\phi_J}-\ket{\phi_J^y}||^2 \leq Df(x) \label{eq:zalka2}
\EE
with the argument
\BE
x=\frac{4J}{D}\sum\limits_{y=1}^D\sum\limits_{j=1}^J ||P_{\mathcal{W}_y}\ket{\phi_j}||^2
\EE
and $P_{\mathcal{W}_y}$ the projector onto the winning space of oracle $y$.The function $f(x)$ is defined in  \cite{Zalka1999} via 
\BE
f\Big(x=4J^2\sin^2\nu\Big)= 4\sin^2(J\nu).
\EE 
Every state $\ket{z}\in \mathcal{H}_A$ is part of the winning space $\mathcal{W}_y$ for exactly $d={{N-1}\choose {n-1}}$ different oracles. As a consequence, the argument $x$ of the function $f$ in \eq{eq:zalka2} is given by
\begin{eqnarray}
x=\frac{4J}{D}\sum\limits_{y=1}^D\sum\limits_{j=1}^J ||P_{\mathcal{W}_y}\ket{\phi_j}||^2
&=&\frac{4J}{D}\sum\limits_{j=1}^J\sum\limits_{y=1}^D \sum\limits_{z\in \mathcal{W}_y}||\braket{z}{\phi_j}||^2\\
& &= 4J\sum\limits_{j=1}^J \frac{d}{D} \sum\limits_{z=1}^N ||\braket{z}{\phi_j}||^2.
\end{eqnarray}
The sum over all states $\ket{z}$ sums up to unity leading to 
\BE
\frac{4J}{D}\sum\limits_{y=1}^D\sum\limits_{j=1}^J ||P_{\mathcal{W}_y}\ket{\phi_j}||^2 = 4J^2\frac{n}{N}=4J^2\sin^2 \nu
\EE
where we used $d/D=n/N$.
\qed
\end{proof}

Grover's algorithm saturates this inequality since we find for this algorithm
\begin{eqnarray}
\ket{\phi}&=&\frac{1}{\sqrt{N}}\sum\limits_{z=1}^N \ket{z}=\cos \nu \ket{\ell}+ \sin \nu\ket{w}\\
\ket{\phi_J^y}&=&\cos [(2J+1)\nu] \ket{\ell}+ \sin [(2J+1)\nu] \ket{w}\\
\ket{\phi_J}&=&\cos \nu \ket{\ell}+ \sin \nu \ket{w}
\end{eqnarray}
leading to
\begin{eqnarray}
\frac{1}{D}\sum\limits_{y=1}^D|| \ket{\phi_J}-\ket{\phi_J^y}||^2&=&2-2 \sin[(2J+1)\nu]\sin(\nu)-2\cos[(2J+1)\nu]\cos(\nu)\\
&=&2[ 1-\cos(2J\nu)].
\end{eqnarray}

The right side of \eq{eq:woelk} is govern by the lemma
\begin{lemma}
The average success probability $p$ to identify any item $z\in \mathcal{W}_y$ out of the winning space $\mathcal{W}_y$ of the oracle $O_y$ with  $1\leq y\leq D$ given the states $\phi^y_J$ average over all oracles is upper bounded by
\BE
2-2\sqrt{p\frac{n}{N}}-2\sqrt{(1-p)\left(1-\frac{n}{N}\right)}\leq \frac{1}{D}\sum\limits_{y=1}^D|| \ket{\phi_J}-\ket{\phi^y_J}||^2
\EE \label{lem:zalka2}
\end{lemma}
\begin{proof}
Again, in order to proof this lemma, we follow the proof in \cite{Zalka1999} and only point out the generalizations we have to make when going form $n=1$ winning state to $n>1$ winning states. Similar to \cite{Zalka1999}, we write the states
\BE
\ket{\phi^y_J}=\sum\limits_{x=1}^X c_x^y\ket{x} \quad \ket{\phi_J}=\sum\limits_{x=1}^X c_x\ket{x}
\EE 
via some orthonormal basis $\lbrace\ket{x}\rbrace$ of some Hilbert space with dimension $X$. The optimal procedure to identify a winning item $\ket{z}$ is to perform projective measurements (see Ref. \cite{Zalka1999}). Let $\lbrace\ket{x}\rbrace$ be the measurement basis and we denote with $X_z$ the subspace containing all states $\ket{x}$ which correctly denote that $\ket{z}$ is a winning item. As a consequence, the success probability $p_y$ if the unknown oracle is given by $O_y$ is determined via
\BE
p_y=\sum\limits_{z\in \mathcal{W}_y}\sum\limits_{x\in X_z}|c_x^y|^2.
\EE
Similar, we can define a success probability $a_y$ for the state \ket{\phi_J} via
(compare Eq.(A7) in \cite{Zalka1999}) 
\BE
a_y=\sum\limits_{z\in \mathcal{W}_y}\sum\limits_{x\in X_z}|c_x|^2.
\EE 

In order to proof \eq{lem:zalka2} Zalka determines the minimal distance  an arbitrary state $\ket{\phi_y}$ with success probability $p_y$ needs to have from a given state  $\ket{\zeta_y}$ with success probability $a_y$. This minimal distance is given (compare Eq. (A8) in \cite{Zalka1999}) by
\BE
||\ket{\phi_y}-\ket{\zeta_y}||^2\geq 2-2\Big(\sqrt{p_ya_y}+\sqrt{(1-p_y)(1-a_y)}\Big).\label{eq:min_difference}
\EE

The minimum of 
\BE
\frac{1}{D}\sum\limits_{y=1}^D|| \ket{\psi_y}-\ket{\zeta_y}||^2
\EE
for all possibly states $\ket{\zeta_y}$ and success probabilities $p_y$ is reached if all $p_y=p$ and $a_y=a$ (see \cite{Zalka1999}). Due to normalization we find
\BE
\sum\limits_{y=1}^D a_y=d\sum\limits_{x=1}^X |c_x|=d
\EE
where we have used that each item $\ket{z}$ belongs to the winning space of $d={{N-1}\choose {n-1}}$ different oracles. As a consequence, the minimum is achieved for $a_y=d/D=n/N$ (see discussion before Eq.(A10) in \cite{Zalka1999}) leading finally to the modification of Eq.(A10) \cite{Zalka1999} to
\BE
\frac{1}{D}\sum\limits_{y=1}^D|| \ket{\phi_J}-\ket{\phi^y_J}||^2 \geq 2-2\sqrt{p\frac{n}{N}}-2\sqrt{(1-p)\left(1-\frac{n}{N}\right)}
\EE 
which gives us directly \lem{lem:zalka2}.
Also this bound is satisfied by Grover's algorithm. 
\end{proof}

The above stated optimality proof of Grover's algorithm can be easily generalized to situation where we start in a state $\ket{\zeta_y}$ with success probability $a_y=a=\sin^2 \phi$  
 and try to optimize the success probability $p_y$ of $V_y^J\ket{\zeta_y}$ with the help of maximal $J$ oracle queries. 
 \lem{lemma:distance} is independent from the initial state and can therefore  directly be applied. 
From \eq{eq:min_difference} we find
\BE
\frac{1}{D}\sum\limits_{y=1}^D|| V_y^J\ket{\zeta_y}-\ket{\zeta_y}||^2\geq 2-2\sum\limits_{y}^D \sqrt{p_y\sin^2\phi}\sqrt{(1-p_y)\cos^2\phi}
\EE
which is minimal if $p_y=p \,\forall y$. Thus we find
\BE
\frac{1}{D}\sum\limits_{y=1}^D|| V_y^J\ket{\zeta_y}-\ket{\zeta_y}||^2\geq 2-2\sqrt{p\sin^2\phi}\sqrt{(1-p)\cos^2\phi}.\label{eq:gGrover1}
\EE

 \lem{lemma:distance} and \eq{eq:gGrover1} can be simultaneously saturated by starting in a state 
\BE
\ket{\zeta_s}=\sin \phi \frac{1}{\sqrt{|\mathcal{W}_y}|}\sum\limits_{\ket{z}\in \mathcal{W}_y}\ket{z} + \cos \phi \frac{1}{\sqrt{|\mathcal{L}_y|}}\sum\limits_{\ket{z}\in \mathcal{L}_y}\ket{z}
\EE
 and performing Grover iterations via the unitary
\begin{eqnarray}
V_y^J&=&\Big[(\mathds{1}-\ket{\psi}\bra{\psi})O_y\Big]^J\\
\ket{\phi}&=&\frac{1}{\sqrt{N}}\sum\limits_{z=1}^N \ket{z}.
\end{eqnarray}
Applying $V^J$ with an empty oracle on $\ket{\zeta_s}$ does not change the success probability $a_y$ leading to a maximal success probability $p=\sin^2(\phi+\nu)$ with $\sin^2\nu=n/N$.

\end{widetext}

\end{appendix}

\end{document}